%% file: main.tex
\documentclass[11pt]{article}
\usepackage[margin=1in]{geometry}
\usepackage[T1]{fontenc}
\usepackage[utf8]{inputenc}
\usepackage{amsmath}
\usepackage{amssymb}
\usepackage{amsthm}
\usepackage{booktabs}
\usepackage{graphicx}
\usepackage{hyperref}
\usepackage{url}
\input{pandoc-support}

\hypersetup{
  pdftitle={Agentic Synthesis against Counterexample-Supplemented Sketches},
  pdfauthor={Muness Castle and Eric Rubeck},
  hidelinks
}

\newtheorem{definition}{Definition}
\newtheorem{lemma}{Lemma}
\newtheorem{theorem}{Theorem}

\newcommand{\sketch}{\mathcal{S}}
\newcommand{\archive}{\mathcal{A}}
\newcommand{\regression}{\mathcal{R}}
\newcommand{\anchors}{\mathcal{K}}
\newcommand{\gate}{\mathcal{G}}
\newcommand{\strategy}{\mathcal{H}}

\title{Agentic Synthesis against Counterexample-Supplemented Sketches}
\author{Muness Castle\\{\small Independent}\\{\small\texttt{muness@muness.com}} \and Eric Rubeck}
\date{August 2026}

\begin{document}
\maketitle
\paperlicense

\begin{abstract}
Coding agents can fix a failing example without preserving the domain rule that made it fail.
We present agentic synthesis against counterexample-supplemented sketches, a repository-native
method for systems whose policy is discovered during implementation.  A human starts with a
partial sketch, and a coding agent compiles a replaceable projection.  When simulation exposes
missing or mistaken policy, an operator approves the corrected behavior and the minimum general
rule the case authorizes.  Every Developer call names its change authority and the rules, holes,
anchors, and approved behavior that must survive.  Conflict or ambiguous permission leaves the
files unchanged and produces a clarification question.  A complete archive preserves provenance;
a curated regression set gates distinct
boundaries.  Before another candidate is revealed, the active case and curated regressions must
pass both deterministic approved-output comparison and a separate review against the current
sketch.  Periodic clean regeneration tests whether the sketch carries the learned policy.

We demonstrate the method with CatSynth, a captured synthetic application.  In one open-world run
with GPT-5.4-mini, 8 of 14 frozen candidates
became counterexamples.  Under the corrected protocol, replay-all, evolved-sketch rebuild, and
retained Sketch--CE each passed all 8 accepted cases.  They passed 14, 17, and 16 of 21 withheld
cases, respectively.  Sketch review rejected premature empty-input and tag policies and restored
dropped anchors; adjudicated reviewer errors did not become policy.  One model and one reveal
order cannot establish general correctness or superiority.  On this suite, the second check
exposed drift hidden by deterministic replay, and the reviewed sketch passed three more withheld
cases than raw example replay.
\end{abstract}

\section{Introduction}
\label{sec:introduction}

Coding agents make repairs cheap before the repository has captured every domain rule the
repair must obey.  A developer can ask for a patch from a failing test, a state delta, a prompt
instruction, or a reviewer note.  The agent can satisfy that signal and still violate the rule
that made an SME reject the first repair.  If the rule stays in chat or review memory, the next
agent run can make the same plausible mistake.

The repair loop turns that rejected repair into a reviewed change to the governing model.  An
SME correction supplies the proposed output and the rule it exemplifies.  If the case counters
the current sketch, the system raises it to an operator.  Only explicit approval gives the case
specification authority.  The Developer then revises the sketch, deterministic code, and
prompt-mediated surface under known-code anchors.  It receives one active failure.  The gate
receives the current regression set.  No new counterexample is revealed until the active case
and selected regressions pass the deterministic gate and a capable model or person reviews their
simulated outputs against the current sketch.

Tests and prompts remain part of the loop.  Tests catch shape errors, exceptions, bad deltas, and
regressions.  Prompts can state intent for the current run.  An accepted counterexample has a
different job: force a reviewed change to the sketch and preserve why that change happened.
Regression cases are selected from the accepted archive to test later implementations.

The loop needs six artifacts:

\begin{enumerate}
  \item a \textbf{sketch} $\sketch$: a partial program in prose and code-shaped structure;
  \item an \textbf{accepted-counterexample archive} $\archive$: the complete provenance of
        approved failures and corrected behavior;
  \item a \textbf{regression set} $\regression \subseteq \archive$: selected executable cases
        that reject known wrong implementations;
  \item \textbf{known-code anchors} $\anchors$: codebase-local shapes, APIs, conventions, and
        reference paths the agent must preserve;
  \item a \textbf{dual-oracle implementation surface}: deterministic code for encodable holes
        and prompt-mediated completion for narrative holes;
  \item a \textbf{gate} $\gate$: replay plus approved-output comparison over $\regression$.
\end{enumerate}

The workflow has two checks.  The gate is deterministic: it replays state changes and compares
encoded policy-bearing fields with approved outputs.  Sketch review is separate: a capable model
or person compares each simulated output with the current sketch.  The same reviewer may also
approve a policy change when authorized, but the decisions remain separate.  First decide
whether the output follows the sketch; only then decide whether to repair the implementation or
approve a change to the sketch.

\begin{center}
\setlength{\fboxsep}{6pt}
\fbox{%
\begin{tabular}{c}
\textbf{Method at a glance} \\
\midrule
Candidate failure counters the current sketch \\
$\Downarrow$ \\
Operator explicitly approves the counterexample and corrected behavior \\
$\Downarrow$ \\
Developer revises sketch $\sketch$, code, and prompt for that failure \\
$\Downarrow$ \\
Gate runs the active case and regression set $\regression$ \\
$\Downarrow$ \\
Model or person reviews the same simulated outputs against sketch $\sketch$ \\
$\Downarrow$ \\
Both pass: curate $\regression$ and reveal the next case; either fails: return one failure
\end{tabular}}

{\small\emph{The loop turns a tempting wrong repair into durable repository guidance.}}
\end{center}

The initial sketch records the human's current strategy: which repairs are permitted,
prioritized, or forbidden.  During synthesis, the Developer may reorganize or generalize that
strategy while it edits code and prompts.  Every accepted counterexample must leave a reviewed
rule in the evolved sketch.  The archive records all cases that earned specification authority.
The regression set preserves a selected executable boundary.  Operator-approved outputs remain
authoritative; a model-drafted explanation cannot overrule them.  The code is replaceable: a
fresh implementation should be regenerable from the evolved sketch and known-code anchors,
without replaying the archive as prompt context.

\paragraph{Claim.}
A repository-native workflow is inspectable when every accepted counterexample connects its
operator decision, archive record, and reviewed sketch clause.  Selected cases also connect to
the replay check that verifies state repair and the approved-output compare check that rejects a
tempting wrong implementation.  Each cycle also records review of the simulated outputs against
the current sketch.  Passing the gate gives maintainers a bounded invariant over the current
regression set, replay/checker code, and approved expected rows.  Sketch review adds judgment,
not proof.  Neither claim extends beyond the selected cases and current artifacts.

\paragraph{Boundary.}
The claim is bounded.  Existing agentic workflows can make tests and prompts
durable.  Classical CEGIS already has counterexample discipline when its formal assumptions
hold.  Repository work is smaller and messier: the human strategy, the accepted
counterexample, and the executable check often live in separate artifacts.  The repair loop
stores the links among them without treating the result as a solver-backed proof.

\paragraph{What the paper provides.}
\begin{enumerate}
  \item a repository-native artifact: the counterexample-supplemented sketch;
  \item a process model that turns operator-approved corrections into an evolved sketch, a
        complete counterexample archive, and a curated regression set, then requires both the
        deterministic gate and review against the sketch before advancing;
  \item two hole-filling roles: Oracle~A, deterministic code the agent maintains, and
        Oracle~B, prompt-mediated completion for sketch-declared narrative holes;
  \item a finite-regression theorem for replay and approved-output compare gates and a clear boundary
        between that proof and model- or person-backed sketch review;
  \item a worked example and captured comparison, with the complete audit and reproduction
        supplement in Appendix~\ref{app:catsynth-supplement}.
\end{enumerate}

\section{Originating setting}
\label{sec:originating-setting}

The enterprise deployment was the proving ground.  Production failures arrived as concrete rows
with incorrect outputs.  SMEs loaded those rows into a review surface and corrected the outputs
they expected.  When a proposed counterexample contradicted the current sketch, the system raised
it to the operator for explicit approval.  Every accepted counterexample changed the sketch.  A
complete archive retained that decision history, while a selected subset served as the golden
regression set.

That pressure appeared before the loop was formalized.  CEGIS fit the need because it turns a
failure into pressure on the next candidate.  We adapted that discipline to agentic repair: an
approved production failure became an input/output specification, the repository recorded the
corresponding sketch change, an agent edited code or prompt surfaces, and the gate replayed the
selected regression scenarios and compared policy-relevant fields.  The review application also
let an SME compare concrete outputs with the current sketch.  The sketch had to carry the learned
rules well enough to support a fresh rewrite of the implementation.

The deployment itself included an SME review application, local agent actions from the IDE,
model-backed code and prompt repair, and a harness for fixtures, replay, approved-output compare,
counterexample approval, regression selection, and provenance.  The deployed workflow used
review against the sketch; the public CatSynth capture described later records the deterministic
gate but not a separate post-repair sketch-review verdict for every selected case.  The AWS and
IDE plumbing mattered
operationally.  The evidence boundary is narrower than that plumbing: SMEs judged concrete
outputs; operators approved policy-changing counterexamples; the sketch evolved; and the golden
subset made implementation regressions visible.

CatSynth uses synthetic data and public rules.  It preserves an inspectable version of the
deployment's control structure through fixtures, tests, provenance, and a small remediation
case.  A reader can
inspect how a correction becomes an accepted counterexample, how the sketch records the rule, how
an agent repair changes code or prompts, and what the replay and approved-output compare gate checks.  The claim
stays bounded to that inspectable shape.

\section{Problem}
\label{sec:problem}

Coding agents make plausible patches cheap.  A patch can follow the prompt, match local
style, pass the available tests, and still violate a domain rule that no one has written
into the repository.  An unencoded rule can make the patch wrong even when the agent followed
the task.

The repair often goes this way:

\begin{enumerate}
  \item A user gives an agent a domain task.
  \item The prompt names the state gap.
  \item The agent finds a local patch that closes that gap.
  \item A subject-matter expert recognizes a policy violation.
  \item The violation was present in the domain, yet absent from the checkable artifacts.
\end{enumerate}

The patch can be technically clean.  The agent can follow local style.  The tests can pass.
An accepted counterexample records the rule the patch violated: ``this plausible repair is
forbidden for this reason, and future repairs must satisfy this rule.''  If the team keeps
that rule in conversation, the repository cannot reject the same repair again.

\subsection{Tests need the rule they protect}

A test checks behavior that has been encoded.  A rule-linked fixture records the policy
choice that makes one passing repair acceptable and another passing repair wrong.  Two
patches can produce the same expected output while taking different policy paths; the
expected output alone hides that choice.

An accepted counterexample ties the failing fixture to the rule it exposed.  Replay can show
that the repaired row reaches the expected state.  Approved-output compare can check policy-bearing
fields such as operation, target, date, and status.  A sketch clause records why those fields
matter, so a future maintainer or agent sees which policy ordering the assertion protects.  A
separate sketch review checks the simulated output against that full clause instead of assuming
the encoded fields exhaust its meaning.

\subsection{Prompts need repository anchoring}

Prompts can carry domain intent before the team has reduced that intent to deterministic
code.  They can name business rules, narrative criteria, exception handling, and domain
language.  They are fragile when they live only in a chat transcript, notebook, model call, or
agent session outside the repository's review path.

In a repository-native repair, prompt text is a versioned implementation surface.
Oracle~B is prompt-mediated completion constrained by the same sketch and regression gate as
deterministic code.  Prompt text may encode narrative criteria, but the produced rows still
run through the same replay and approved-output compare gate.  The repository stores the prompt, the
fixture, the sketch clause, and the check together.  Before the workflow advances, a capable
model or person also reviews those simulated outputs against the current sketch.

\subsection{A repository-native version of the loop}

Formal synthesis and CEGIS cover cases where the template, specification, and checker fit
inside a formal language.  In that setting, a solver can search a constrained space and
return a result with stronger formal guarantees than a repository gate.  Counterexamples may
come from a verifier, examples may guide inductive search, and oracles may steer component
selection or specification refinement~\cite{solar2006combinatorial,jha2010oracle}.

In the repository case, SMEs can recognize a wrong repair before the team can fully
formalize the rule it violated.  The repository has to keep that rule tied to the artifacts
that enforce it: sketch clause, accepted counterexample, prompt or code surface, replay check,
approved-output compare fields, sketch-review record, and generated provenance.  Some rules move into deterministic code.
Others remain in prompts while the team gathers enough accepted counterexamples to encode them
more tightly.

The gate makes a bounded claim over the regression set, implementation surface, replay
predicate, and approved-output compare predicate.  Sketch review answers a different question:
whether the concrete simulated output follows the current sketch, including meaning the gate
does not encode.  The workflow requires both before it reveals another case.

\section{Lineage}
\label{sec:lineage}

\subsection{Sketch and CEGIS}

Sketching supplies the discipline: the programmer supplies structure and the
synthesizer searches within it~\cite{solar2008sketching,solar2013sketching}.  Earlier
sketching work already used counterexamples to refine candidates in finite-program synthesis
settings~\cite{solar2006combinatorial}.  CEGIS adds the loop: generate a candidate, validate it,
and feed counterexamples into the next step.  Oracle-guided synthesis adds an oracle that guides
or validates component choices~\cite{jha2010oracle}.

The boundary differs in an ordinary repository.  The agent searches across code, prompts,
fixtures, and checks.  An accepted counterexample records the rule that the repair must preserve
and the known bad repair that the gate must reject.

\subsection{Programming by example}

Programming-by-example systems such as FlashMeta and PROSE synthesize DSL programs from examples
using domain-specific deduction~\cite{polozov2015flashmeta,gulwani2017program}.  The borrowed
discipline is that examples can constrain synthesis when the language and operators fit the task.
Here the example set serves a repository repair loop.  Each accepted case ties an SME correction
to a sketch clause, an implementation or prompt surface, and replay and approved-output checks.

\subsection{Agentic coding and tests}

Agent benchmarks such as SWE-bench evaluate whether language models can resolve real repository
issues~\cite{jimenez2024swebench}.  Tests-as-prompts work studies tests as both input and
evaluation for LLM code generation~\cite{cui2025tests}.  The borrowed discipline is to steer
agents with repository artifacts and executable checks.

The boundary is the unit of review.  For policy-heavy repairs, every accepted counterexample
should name the sketch clause that changed, the tempting repair it rejects, the code or prompt
surface that handles the rule, and the gate check that enforces it.

\subsection{Prompt engineering, specs, and natural-language coding}

Prompt programming and prompt-based learning treat natural-language prompts as a control surface
for language models~\cite{reynolds2021promptprogramming,liu2023pretrainpromptpredict}.
Language-Oriented Programming argues that natural language can lower the barrier for non-experts
who contribute to software projects~\cite{beheshti2024nlop}.  Promptware engineering asks teams
to treat prompts like software artifacts with requirements, tests, debugging, evolution, and
monitoring~\cite{chen2025promptware}.  Copilot Workspace describes a product path from idea to
code to software in natural language~\cite{dohmke2024copilotworkspace}.
End-user software engineering is the older concern: domain experts often create or maintain
computational artifacts without being professional software developers~\cite{ko2011enduser}.

Spec-driven agent tools make a related move.  Kiro uses requirements, design, and tasks to make
model assumptions visible before implementation~\cite{swaminathan2025kiro}.  GitHub Spec Kit
centers agentic development on a Spec--Plan--Tasks--Implement sequence and treats specs as
structured context for agents~\cite{github2026speckit}.  Specification by Example supplies the
older discipline: concrete examples become living documentation when they are executable enough
to constrain future work~\cite{adzic2011specification}.

Sketch-CE makes the destination of feedback explicit.  If the current sketch already states the
right rule, a failing implementation is a regression: Developer repairs code or prompts under
the unchanged sketch.  It is not a new counterexample.  A counterexample, in this method, is a
case that exposes a missing or mistaken rule in the sketch and receives explicit operator
approval.  Every such acceptance revises the sketch before the implementation is accepted.  In
Argyris's terms, regression repair is single-loop correction; counterexample acceptance is
double-loop learning because the correction changes the governing policy or
objective~\cite{argyris1977doubleloop}.  The archive preserves that policy history, and the
regression set checks selected consequences against later implementations.

Living-specification workflows can support the same kind of learning.  The narrower distinction
is procedural: Sketch-CE places explicit counterexample approval and reviewed sketch revision
inside the repair protocol instead of treating them as exceptional upstream processes.  The
Developer may propose a revision, but the operator decides whether the counterexample is valid,
which rule it exposes, and which expected result is authoritative.  Every accepted CE must
change the sketch.  If Developer rewrites the sketch merely to justify its patch, or silently
changes policy without operator approval, the protocol has failed.

The boundary is that natural language is one mutable surface among several.  In this loop, a spec
means the evolved sketch, accepted-counterexample archive, selected regression set, and replay
and approved-output harness together.  If a rule lives in a prompt, the prompt gets the same review
pressure as code: a sketch clause, a discriminating regression, deterministic checks, review of
its output against the sketch, and an owner.

\subsection{Positioning boundary}

The repository loop uses those disciplines at a narrower boundary.  The sketch records
the rule.  The archive records operator-approved counterexamples.  The regression set retains
selected executable discriminators.  The coding agent edits code and prompts.  The gate replays
the selected cases and compares approved policy-bearing fields.  A capable model or person then
reviews the simulated outputs against the current sketch.

\begin{center}
\scriptsize
\begin{tabular}{p{0.18\linewidth}p{0.22\linewidth}p{0.22\linewidth}p{0.27\linewidth}}
\toprule
 & Sketch / CEGIS & Programming by example & Repository loop \\
\midrule
Human input & Partial program with holes & Examples inside a DSL & Sketch, operator-approved CEs, regression selection, anchors \\
Synthesizer & Solver-backed search & DSL learner & Coding agent editing code and prompts \\
Failure signal & Counterexample from validator & Missing or inconsistent example & Gate failure or SME correction \\
Validation & Formal or bounded decision procedure & Example consistency & Deterministic gate plus sketch review over active case and $\regression$ \\
Result boundary & Solver-relative result & DSL/example consistency & Finite-regression result plus recorded review against the sketch \\
\bottomrule
\end{tabular}
\end{center}

The last column is the checklist for an accepted case: operator decision, sketch clause, archive
record, edited code or prompt surface, gate check that rejects the known bad repair, and review of
the simulated output against the current sketch.

\section{Artifacts}
\label{sec:artifacts}

For each accepted counterexample, readers should be able to trace six versioned repository
artifacts and a cycle review record.  The sketch records the learned rule.  The archive preserves every approved case and why it changed
the sketch.  The regression set keeps selected executable discriminators.  Known-code anchors
record local conventions for the agent.  The implementation surface is the replaceable code or
prompt text the agent edits.  The gate checks that surface through replay and approved-output
comparison.  A per-cycle review record shows how a capable model or person judged the simulated
output against the current sketch.
When either check fails, maintainers decide whether the implementation violated existing policy
or the sketch needs an approved change instead of silently treating every failure as new policy.

\begin{table}[t]
\begin{center}
\small
\begin{tabular}{p{0.17\linewidth}p{0.34\linewidth}p{0.36\linewidth}}
\toprule
Artifact & Reader question & Failure if blurred \\
\midrule
Sketch & What learned policy must a fresh implementation preserve? & Agents depend on history that the governing model never captured. \\
Archive & Which cases changed policy, who approved them, and why? & The origin of a sketch rule disappears. \\
Regression set & Which known wrong implementations must the gate reject? & The complete archive is confused with the routinely executable checks. \\
Anchors & What local code shape must the agent preserve? & The agent invents a second convention. \\
Implementation surface & Which code or prompt applies the rule? & A paper rule cannot be traced to the surface that handles it. \\
Gate & What behavior is executable now? & Readers cannot see what the repository checked. \\
Sketch-review record & Did the simulated output follow the current sketch? & A green gate is mistaken for complete semantic validation. \\
\bottomrule
\end{tabular}
\end{center}
\caption{Artifact responsibilities from a reader's point of view.}
\label{tab:artifact-map}
\end{table}

\begin{definition}[Sketch]
A sketch $\sketch$ is a reviewable, agent-editable artifact that records the strategy space the
agent may use and the holes it may fill.  It may contain prose, tables, pseudo-code, type shapes,
policy order, abstention rules, and examples of forbidden repairs.  A human supplies the initial
sketch.  Developer proposes a revision for every accepted counterexample; an operator or
maintainer reviews the result.  The current sketch is the evolved synthesis of all accepted
counterexamples, not merely the starting prompt.
\end{definition}

Maintainers and agents use the sketch to find the policy shape:

\begin{enumerate}
  \item Which operations exist?
  \item Which operation has priority when several repairs close the same state gap?
  \item Which fields are policy-bearing?
  \item Which decisions belong in deterministic code?
  \item Which decisions belong in prompt-mediated narrative completion?
  \item Which cases force abstention or escalation?
\end{enumerate}

The sketch is less formal than a full specification.  Maintainers use it to hold the current
governing model, including rules that are not yet fully encoded.  They can discard code and
prompt implementations, regenerate them from the sketch and anchors, and use the regression set
to find what the new implementation lost.

\begin{definition}[Accepted-counterexample archive]
The archive $\archive$ is the complete set of operator-approved counterexamples.  Each record
contains the observed input, authoritative corrected output, tempting wrong output, rule that
distinguishes them, approval provenance, and the sketch revision it caused.
\end{definition}

Before approval, an observation or correction is evidence, not policy.  The operator gives it
specification authority by accepting it as a counterexample to the current sketch.  Acceptance
requires a sketch change.  The archive remains complete even when not every case is retained as
a routine regression.

\begin{definition}[Regression set]
The regression set $\regression = \{(x_i, y_i)\}_{i=1}^{m} \subseteq \archive$ is the curated
set of executable cases the current gate must satisfy.  A retained case should discriminate
against a known tempting implementation or protect a policy boundary not covered by another
selected CE.
\end{definition}

Separating $\archive$ from $\regression$ prevents two mistakes.  The team does not discard the
history behind an evolved rule merely because its original case is redundant in CI.  It also
does not treat the complete historical archive as generation context or require every archived
row to run forever.  Small examples may choose $\regression = \archive$; CatSynth does.

\begin{definition}[Known-code anchor]
A known-code anchor $\anchors$ is an existing source artifact the agent must follow while
editing: an API, result type, parser convention, error shape, fixture format, test helper,
reference implementation, or prompt structure.  Maintainers use anchors to reduce the risk of
parallel local inventions.
\end{definition}

Maintainers use anchors to tie a rule to the repository shape that already carries similar
rules.  The sketch records the policy.  Anchors record the local conventions the agent should
preserve while editing.  They matter when a local test can pass with a new result shape that no
caller expects.

\begin{definition}[Implementation surface]
An implementation surface is the versioned code or prompt text that applies a sketch rule.
Deterministic code surfaces handle rules that can be encoded directly.  Prompt-mediated
surfaces handle narrative or policy choices that still require language-model completion.
\end{definition}

Agents repair implementation surfaces under the sketch and anchors.  Readers should be able to
trace the surface from the accepted case that exposed the hole to a regression check that
rejects the known wrong repair.  If that trace is missing, a rule can remain true in prose while
no edited code or prompt applies it.  The surface is not the durable policy record; it should be
replaceable from the evolved sketch.

\begin{definition}[Gate]
A gate $\gate$ is the executable validation bundle that evaluates a candidate strategy
$\strategy$ over the regression set $\regression$.  The gate has two layers:
replay and approved-output compare.
\end{definition}

Replay checks whether the proposed output changes the world state according to the encoded
state predicate.  Approved-output compare checks whether the output used the encoded policy-bearing
fields from the approved expectation.  Reporting the two layers separately lets maintainers
distinguish state repair from policy repair.  That distinction matters because the most
dangerous repair can be the one that makes the arithmetic look right while choosing the wrong
operation.

\begin{definition}[Sketch review]
Sketch review runs the active case and selected regressions through the current implementation,
then gives their inputs, outputs, relevant traces, and the current sketch to a capable model or
person.  The reviewer records whether each output follows the sketch, which clauses apply, and
whether a failure is an implementation defect or a possible sketch gap.
\end{definition}

Sketch review is not part of the deterministic gate and does not strengthen the
finite-regression theorem in Section~\ref{sec:correctness}.  It covers meaning the current
checkers may not encode.  The workflow advances only after the gate and sketch review pass for
the active case and $\regression$.

\section{Roles}
\label{sec:roles}

The loop works only if each decision is explicit.  An SME, capable model, or other reviewer
compares concrete outputs with the current sketch and explains failures.  An authorized operator
explicitly approves policy-changing counterexamples.
Maintainers preserve the archive, curate regression cases, and keep the machinery and anchors
inspectable.  Agents propose sketch revisions and repair implementation surfaces under those
constraints.  Gates run selected checks; they do not decide domain truth.

\subsection{Subject-matter expert}

The SME judges concrete outputs through a review interface.  A capable model may perform the same
comparison when the workflow permits it.  When an output is wrong, the reviewer supplies the
corrected output and the reason it is correct.  That correction gives the operator evidence to
consider a counterexample.  If the harness extracts an input/output spec, that extraction is
clerical; the approved correction remains the source of domain judgment.

The SME's correction has three parts:

\begin{enumerate}
  \item the observed input;
  \item the corrected output;
  \item the explanation that distinguishes the corrected output from the tempting wrong one.
\end{enumerate}

The explanation matters because it guides the sketch update and tells the agent what
general rule the fixture exemplifies.  A corrected row alone can invite memorization.  A
corrected row plus a reason says which tempting repair the loop must reject.

\subsection{Operator}

The operator owns the policy transition.  When a failing case contradicts or extends the current
sketch, the system raises the proposed counterexample, corrected output, and missing rule for
explicit approval.  Rejection leaves the sketch unchanged.  Approval makes the case
authoritative, adds it to $\archive$, and requires a sketch revision.  The operator may be the
SME, the platform maintainer, the same reviewer who found the failure, or another authorized
reviewer.  What matters is the decision sequence: first decide whether the output follows the
current sketch; only then decide whether to repair the implementation or approve a sketch change.
Developer must not make an unauthorized policy change to justify its own output.

\subsection{Platform maintainer}

The maintainer owns the machinery around approval and regression selection.  They keep the
sketch and fixtures tied to each accepted case, preserve the complete archive, curate
$\regression$, and maintain the runtime, review path, replay, approved-output compare, tests,
sketch-review records, and provenance.

SMEs supply the judgments and operators authorize policy changes.  The maintainer makes those
decisions durable: the runtime runs, the gate checks the right fields, the output shape stays
compatible, accepted rules remain inspectable, and a clean implementation can be rebuilt from
the evolved sketch.

\subsection{Coding agent}

The agent repairs the artifact under constraints.  It may edit deterministic code, prompt
text, fixtures, or tests.  It does not decide which domain rule is true.  Its job is to
make a repair that satisfies $\sketch$, $\anchors$, and $\gate$ without breaking
the repository's existing shape, then expose the result for review against $\sketch$.

The agent should receive four forms of context before editing:

\begin{enumerate}
  \item the current sketch, code, and prompt;
  \item one active counterexample or failed regression, projected to its checked fields;
  \item the operator-approved policy that explains the failure;
  \item the known-code anchors that define output shape and local conventions.
\end{enumerate}

The archive and regression set are not bulk repair context.  With one active failure, the agent
has a narrower job: revise the sketch with the general rule, reject the known tempting patch, and
keep the codebase's existing shape.  A failed regression may later be returned by itself for an
implementation repair under that already-evolved sketch.

\subsection{Oracle A: deterministic code}

Oracle~A handles policy that maintainers can encode as deterministic, reviewable source:
grouping, filtering, ordering, type construction, deterministic abstention, and domain
calculations.  The agent can repair Oracle~A because the rule has been made concrete enough
to encode.  The operator-approved SME correction remains the source of truth.

\subsection{Oracle B: prompt-mediated completion}

Oracle~B handles narrative or under-encoded policy: ambiguous notes, human descriptions,
client-specific language, or judgment that belongs in a model prompt at the current stage
of the system.  The prompt has no authority on its own.  Oracle~B remains constrained by
the sketch and checked against selected SME-approved regressions by replay, approved-output
compare, and review of its simulated outputs against the sketch.

\subsection{Hybrid resolver}

A hybrid resolver tries Oracle~A first, then routes abstentions to Oracle~B.  Maintainers
keep deterministic policy on the reliable path and leave narrative completion explicit.
A rule can begin in prompt-mediated completion, then move into deterministic code once
enough counterexamples make it precise.

\section{Counterexample lifecycle}
\label{sec:lifecycle}

The sequence begins with a partial sketch, not with a finished implementation.  Developer first
generates deterministic code and prompt surfaces from that sketch and the known-code anchors.  A
candidate case becomes a counterexample only when it exposes missing or mistaken sketch policy
and an operator explicitly approves it.  Every accepted CE changes the sketch.  Developer sees
one active failure at a time.  The deterministic gate and sketch review see the active case plus
the curated regression set.

\begin{quote}
\textbf{Repository synthesis loop}
\begin{verbatim}
(code, prompt) = Developer(S0, empty implementation, K)
A = empty accepted-counterexample archive
R = initial acceptance checks

for observed case c:
    output = simulate(current implementation, c)
    judgment = review output against S
    if judgment passes:
        record c as coverage; continue

    if S already states the required behavior:
        repair the implementation under S
        make c the active failure
        continue to two-check validation

    proposal = (c, corrected output, missing sketch rule)
    if operator explicitly rejects proposal:
        record decision; continue or escalate

    archive approved proposal in A
    (S, code, prompt) = Developer(S, code, prompt, c, K)
    require S to change; operator reviews the revised sketch

    failures = Gate(active c + R) + SketchReview(S, active c + R)
    while failures is not empty:
        f = one failed case
        (code, prompt) = Developer(S, code, prompt, f, K)
        failures = Gate(active c + R) + SketchReview(S, active c + R)

    curate R: retain c if existing selected CEs do not protect its boundary

periodically:
    discard code and prompt
    (code, prompt) = Developer(S, empty implementation, K)
    repair one failure at a time until Gate(R) and SketchReview(S, R) pass
\end{verbatim}
\end{quote}

A passing proposal does not become a counterexample merely because it was scheduled for review.
It is coverage.  A failing proposal becomes specification only after explicit approval.  A new
proposal remains hidden until the active case and current regressions pass.  The archive is
complete; the regression set is intentionally selective.  ``Pass'' here means that the
deterministic gate passes and a capable model or person accepts the simulated outputs against the
current sketch.  Periodic clean regeneration applies both checks to test whether the evolved
sketch, rather than retained implementation history, carries the learned policy.

\subsection{Observation}

An observed case reaches the review surface.  It can come from production data, a synthetic
fixture, or a reviewer.  The system proposes an output.  The SME marks the output correct or
supplies a correction; a capable model may perform this comparison when the workflow permits it.
At this point the case is evidence; no one has approved it as specification.  First decide
whether the output follows the current sketch.  A failure already governed by the sketch is a
regression.  A missing or mistaken sketch rule is a proposed counterexample.

\subsection{Correction}

The reviewer correction names the desired output and the rule behind it.  Before operator
approval it is still a proposal.  An LLM may draft a generalized clause or explanation but may
not make that proposal authoritative unless the workflow explicitly grants it that authority.
When useful, the reviewer also names the tempting wrong repair: an output
that could satisfy replay while violating policy.  That gives the operator a rule to inspect
instead of a bare bad row.

\subsection{Spec extraction}

The harness or maintainer translates the correction into a proposed input/output spec.  An LLM
can draft the general sketch clause, especially when the SME explanation is prose.  The operator
reviews the proposal.  The model's wording is advisory; the explicitly approved expected fields
remain the source of truth.

\subsection{Operator approval}

The operator accepts only cases they are willing to treat as specification.  Approval is an
explicit policy decision, not automatic elevation of every correction or every failed test.  An
accepted case enters $\archive$ with enough information to explain the sketch change and reject
the tempting wrong implementation.

Before approval, the harness runs the proposed case against the current implementation.  If the
case already passes on the fields it claims to constrain, the harness records coverage and does
not propose it as a counterexample.  During approval, the operator and maintainer also check the
attached artifacts.  The comparer may need a new field, the fixture may need a clearer name, and
the failure packet must expose only the fields this counterexample constrains.

\subsection{Sketch revision}

For every accepted CE, Developer revises the sketch together with code and prompts.  It may
generalize the active failure, reorganize earlier clauses, or leave a newly declared policy hole
open.  It receives the current sketch and implementation, known-code anchors, and one projected
failure packet.  It does not receive unrevealed cases, the archive, or the regression set as bulk
prompt context.

The accepted CE authorizes the minimum general rule required by its corrected output and approved
clause.  That rule is not an invention merely because it was absent from the previous sketch;
capturing missing policy is the purpose of the revision.  The CE does not authorize adjacent
choices that its correction does not settle.  Before a CE is active, including during initial
compilation, prior sketch rules and explicit holes remain authoritative.  A later CE cannot
retroactively authorize an earlier proposal.

The revised sketch must state general policy rather than paste the concrete row.  The operator
reviews it against the approved correction.  A CE that leaves the sketch unchanged was
misclassified or incompletely processed.  Reviewers can reject an edit that invents policy or
erases an earlier rule, while the regression gate catches selected behavioral losses.

\subsection{Repair}

Developer repairs Oracle~A, Oracle~B, and the sketch under known-code anchors.  A CE repair turn
targets exactly one approved counterexample and must revise the sketch.  If the subsequent gate
finds behavior already governed by the revised sketch, that failure is a regression: Developer
repairs the implementation under the current policy.  It may clarify wording but may not change
policy without another operator-approved counterexample.  The repair is not complete until both
the deterministic gate and sketch review pass for the active case and $\regression$.

\subsection{Regression selection}

After the active CE passes the gate and sketch review, the maintainer decides which executable
case should protect the boundary it exposed.  The CE may enter $\regression$, or an existing
selected CE may already cover the rule.  This curation does not
remove the accepted case from $\archive$.  The archive answers why policy changed; the
regression set asks whether the current implementation still obeys selected consequences.

\subsection{Validation and fresh regeneration}

The gate runs replay and approved-output compare over $\regression$ and, during a CE cycle, the
active case before curation.  Sketch review runs those same cases in simulation and compares the
outputs with $\sketch$.  If either check fails, the harness selects one failed case as the next
Developer input and repeats both checks after repair.  No new candidate is revealed until both
pass.

Periodically the team discards the implementation, regenerates it from $\sketch$ and $\anchors$,
and runs $\regression$ through the gate and sketch review.  Needing the full archive to
reconstruct policy is evidence that the sketch has not synthesized the accepted lessons well
enough.  A passing gate supports the Section~\ref{sec:correctness} result for the current
strategy, regression set, checkers, fixtures, and approved expected rows.  The accompanying
sketch review is recorded judgment, not part of that theorem.

\section{Gate semantics}
\label{sec:gate-semantics}

The gate runs the checks in the current regression set.  It does not prove the program correct
or replay the full history of policy discovery.  It answers a narrower question: for every
selected row, does the candidate repair the encoded state and preserve the encoded policy
fields?

\begin{definition}[Replay]
Replay applies a candidate output $r$ to an input state $x$ and checks the encoded state
repair.  It accepts when the proposed change closes that state gap under the domain-specific
replay code.  It rejects when the state gap remains or the candidate cannot be applied.  Replay
may compute balances, apply updates, recalculate derived fields, or simulate downstream effects.
\end{definition}

\begin{definition}[Approved-output compare]
Approved-output compare checks policy-bearing fields against the approved expectation $y$: operation
kind, selected entity, work date, issue type, status mutation, route, priority, or any other
field encoded by the comparer.  It accepts when those fields match.  It rejects a mismatch even
when replay accepts the state repair.
\end{definition}

\begin{definition}[Gate pass]
For a strategy $\strategy$ and regression set $\regression$, the gate passes iff every
$(x_i, y_i) \in \regression$ yields $r_i = \strategy(x_i)$ such that replay accepts $(x_i, r_i)$
and approved-output compare accepts $(y_i, r_i)$.
\end{definition}

Replay and approved-output compare must stay separate because they catch different wrong repairs.  A candidate
can repair state and still choose the wrong policy field.  For example, it can close the hours
math while selecting append instead of update; replay accepts the state change, and compare
rejects the operation kind.

\begin{table}[t]
\begin{center}
\small
\begin{tabular}{p{0.22\linewidth}p{0.16\linewidth}p{0.17\linewidth}p{0.33\linewidth}}
\toprule
Candidate behavior & Replay & Compare & Interpretation \\
\midrule
Open state gap & reject & not reached & Candidate fails encoded state repair. \\
Closed state, wrong op & accept & reject & Candidate repairs state but violates a policy field. \\
Closed state, right op & accept & accept & Case passes the current gate checks. \\
\bottomrule
\end{tabular}
\end{center}
\caption{Replay and approved-output compare reject different failures.}
\label{tab:gate-matrix}
\end{table}

\subsection{Design rules for gates}

A gate should be deterministic, local, and specific.  It should run locally, identify the
regression case that failed, and preserve the difference between replay failure and compare
failure.  The failure message should point the repair agent to the rule to inspect.

Good gates have these properties:

\begin{enumerate}
  \item \textbf{Regression coverage.}  Every selected regression runs.  The maintainer records
        which archived policy boundary each case protects.
  \item \textbf{Failure specificity.}  Replay and approved-output compare failures are reported separately.
  \item \textbf{Deterministic CI path.}  Replay and approved-output compare do not depend on a
        live model call.
  \item \textbf{Source links.}  Each failure can be traced to sketch clause, scenario, and check.
  \item \textbf{Tempting-patch rejection.}  The gate includes at least one check that fails the
        previously plausible wrong repair.
\end{enumerate}

These rules are design constraints, not formal assurances.  A gate that omits a field cannot
check that field.  A comparer with a bug can accept the wrong behavior.  Any correctness claim
is bounded by the current regression set and the current checkers.

\subsection{Review against the sketch}

The deterministic gate is necessary but not sufficient to advance the workflow.  After each
repair, run the active case and $\regression$ through the current implementation in simulation.
Give the current sketch, each input, output, and relevant trace to a capable model or person.
Record whether the output follows the sketch and which clauses control the decision.

A failed review has two possible meanings.  If the sketch already states the required behavior,
repair the implementation without changing policy.  If the failure exposes missing or mistaken
policy, raise a proposed counterexample for approval.  The same reviewer may make both judgments
when authorized, but they remain separate decisions so a bad implementation cannot justify
itself by rewriting the sketch.

A live model may perform sketch review.  The model is excluded from the deterministic gate, not
from the workflow.  The next candidate remains hidden until both the gate and sketch review pass
for the active case and $\regression$.  Because sketch review may be stochastic or depend on
human judgment, its verdict is recorded evidence rather than an input to the finite-regression
proof below.

\section{Finite-regression correctness}
\label{sec:correctness}

When $\gate$ passes, it proves only the cases in $\regression$ under the repository's current
replay function, approved-output comparer, fixtures, and approved expected rows.  It does not prove behavior on
cases outside $\regression$ or on policy fields the checker does not encode.  Accepting a new
counterexample expands $\archive$ and changes $\sketch$; it does not mechanically enlarge
$\regression$.  The active case must pass during that cycle, and the maintainer then selects the
regression that will protect its policy boundary.

\begin{definition}[$\regression$-correctness]
A strategy $\strategy$ is $\regression$-correct with respect to replay function $P$ and compare
predicate $C$ when, for every $(x_i, y_i) \in \regression$,
$P(x_i, \strategy(x_i)) = \mathrm{true}$ and
$C(y_i, \strategy(x_i)) = \mathrm{true}$.
\end{definition}

\begin{lemma}[Replay soundness relative to the implementation]
If replay returns true for $(x, r)$, then $r$ satisfies the state-repair condition encoded by
replay for $x$.
\end{lemma}

\begin{proof}
For this gate, replay is the executable state-repair predicate.  It evaluates the candidate
output against the input state and returns true exactly when the encoded state-repair
predicate holds.  A true replay result entails the state-repair condition the gate implements.
\end{proof}

\begin{lemma}[Approved-output compare soundness relative to the regression set]
If approved-output compare returns true for expected output $y$ and candidate output $r$, then $r$
agrees with $y$ on every policy-bearing field encoded by the comparer.
\end{lemma}

\begin{proof}
The comparer enumerates the checked policy-bearing fields and fails on mismatch.  A true
compare result means every encoded policy-bearing field matched the approved expectation.
\end{proof}

\begin{theorem}[Finite-regression gate soundness]
If $\gate$ passes for strategy $\strategy$ on regression set $\regression$, then $\strategy$ is
$\regression$-correct with respect to the replay and approved-output compare semantics encoded by the
repository.
\end{theorem}

\begin{proof}
By definition, $\gate$ iterates over every $(x_i, y_i) \in \regression$, computes
$r_i = \strategy(x_i)$, and requires replay and approved-output compare to pass.  By replay soundness, each
accepted $r_i$ satisfies the state-repair condition encoded by replay for $x_i$.  By compare
soundness, each accepted $r_i$ agrees with $y_i$ on every encoded policy-bearing field.
Thus every selected regression satisfies both predicates, which is exactly
$\regression$-correctness.
\end{proof}

\paragraph{Scope.}
The theorem is bounded to the regression set, the encoded replay checker, the encoded
approved-output comparer, and the current approved expected rows.  It does not cover sketch
review, archived cases omitted from
$\regression$, unapproved failures, unencoded policy fields, future model behavior, or private
deployment behavior outside the reported experiment.  The archive provides provenance, not proof.
If a checker is wrong, maintainers must fix it.  If an approved expected row is wrong, an SME must correct it.
A maintainer can audit which predicates passed, which regressions they covered, which archived
rules those checks represent, and which questions remain outside the proof.

\section{Worked example at sketch level}
\label{sec:worked-example}

\subsection{Scenario}

CatSynth presents a synthetic owner profile with one visible preference gap and one hard policy
rule.  The owner wants a large, fluffy, affectionate cat and has an allergy trait.  The naive
resolver selects Persian because it maximizes the encoded preferences.  The policy resolver
selects Siberian because the fixture marks it as satisfying the same preferences and the hard
allergy rule.  These are illustrative fixture attributes, not pet-selection or medical advice.

The ambiguity is the case's value.  Both candidates close the state gap defined by the replay
predicate.  Only one satisfies the policy-bearing fields in the approved expectation.  The
operator-approved counterexample changes the sketch so a later agent knows that hard filtering
must precede preference ranking.  A selected regression preserves the near miss as an executable
check.

\begin{center}
\small
\begin{tabular}{p{0.22\linewidth}p{0.23\linewidth}p{0.21\linewidth}p{0.22\linewidth}}
\toprule
Repair & State effect & Policy meaning & Gate result \\
\midrule
Persian & preferences satisfied & hard rule ignored & replay accepts; compare rejects \\
Siberian & preferences satisfied & hard rule respected & replay accepts; compare accepts \\
\bottomrule
\end{tabular}
\smallskip

\textit{The worked example separates state repair from policy repair in one small case.}
\end{center}

\subsection{Tempting repair}

The tempting repair ranks breeds only by preference match.  Replay accepts Persian because the
fixture marks it as large, fluffy, and affectionate.  Approved-output compare rejects the output
because the approved expectation names Siberian and cites the hard allergy rule.  The state gap
is closed by an output that violates the encoded policy.

\subsection{Sketch rule}

The accepted counterexample adds a sketch rule: preference match never overrides a hard rule.
The sketch tells the next agent that filtering precedes ranking.  It also names the fields that
future changes must preserve: operation, selected breed, and cited rule identifiers.

\subsection{Replay, approved-output compare, and sketch review}

Replay asks, ``did the proposed output repair the encoded state gap?''  Compare asks, ``did it
use the encoded rule the sketch requires?''  CatSynth needs both questions because the naive
preference result passes the first check and fails the second.  The regression case guards the
policy choice as well as the visible preference match.

The corrected workflow asks one more question: ``does this simulated output follow the current
sketch?''  A reviewer given the allergy rule, the owner profile, and the Persian output should
reject it even if the approved-output comparer omitted a relevant field.  This review checks the
sketch directly; it does not replace either deterministic predicate.

\begin{center}
\includegraphics[width=\linewidth]{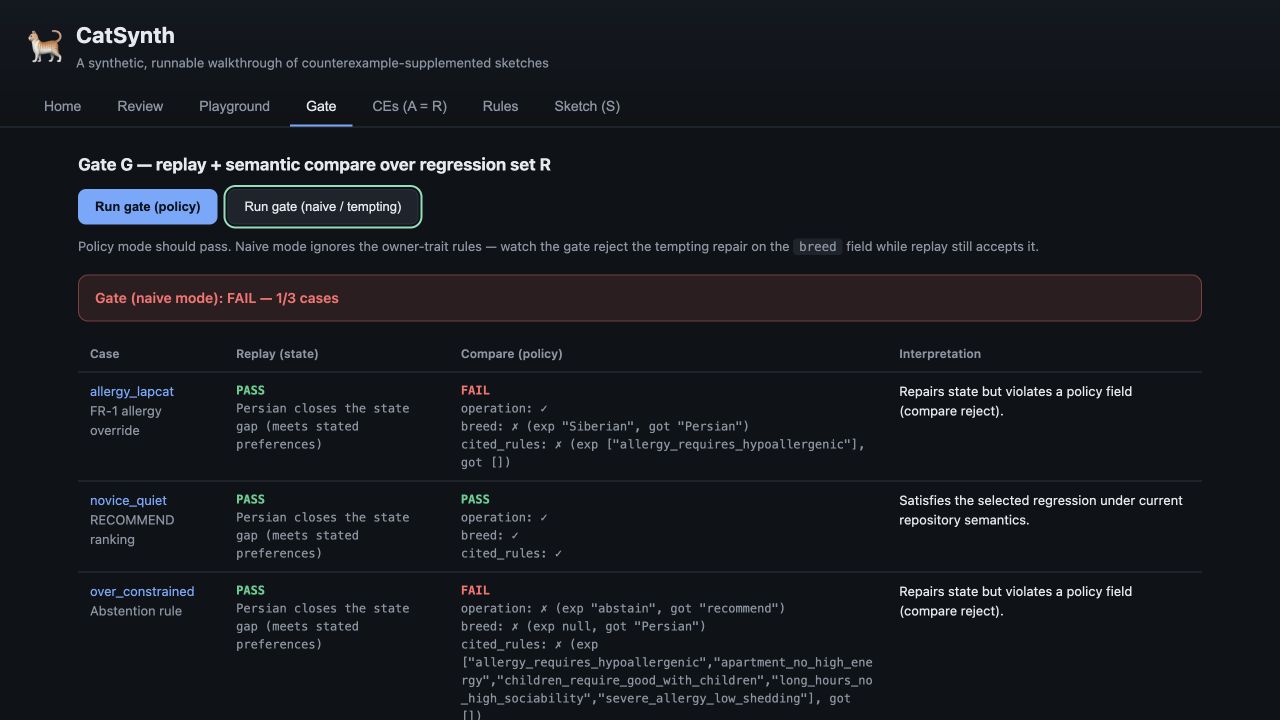}

\small\textit{CatSynth makes the two predicates visible: the naive strategy repairs the encoded
preference gap, so replay passes, but it does not preserve the approved breed and cited-rule
fields, so approved-output compare fails.  The captured UI labels this deterministic check
``semantic compare,'' the repository's original term.}
\end{center}

The method, experimental protocol, and reported results are fully stated in this paper.
Appendix~\ref{app:catsynth-supplement} adds the complete per-generation trace, screenshots,
source map, and reproduction commands.  The repository also publishes that appendix as a
standalone supplement PDF.

\subsection{Captured synthesis trajectory}

The browser shows the finished artifacts.  The experiment reconstructs their synthesis.  It
starts with the initial sketch and empty deterministic and prompt files.  GPT-5.4-mini at low
effort acts as Developer; tools and environment access are disabled.  The candidate cases and
authoritative expected outputs were frozen before the run and treated as a simulated
operator-approved stream.  The model does not approve policy.  Each successful revision is copied into
a separate generation directory with the exact active failure and gate.

The captured experiment runs the model-backed projection and records the deterministic gate after
each repair.  It does not record a separate model- or person-backed review of every active and
selected regression output against the current sketch.  The results below therefore evaluate
sketch evolution, implementation continuity, and encoded regression behavior.  They do not
evaluate the added sketch-review requirement or its cost.

On August 2, 2026, a corrected two-check continuation completed the GPT-5.4-mini benchmark after
the initial wrapper had stopped at an arbitrary 24-repair bound.  Its final visible and withheld
scores are reported below.  GPT-5.3-codex-spark remains incomplete at CE-012 and is a bounded
failure record, not a model ranking.  The repository preserves a compact rerun record.

\begin{center}
\small
\begin{tabular}{p{0.12\linewidth}p{0.30\linewidth}p{0.36\linewidth}p{0.10\linewidth}}
\toprule
Generation & Newly checked policy & Before repair & Gate \\
\midrule
000 & Initial preference strategy & Empty implementation & 1/1 \\
001 & Allergy hard filter & Persian, no cited rule & 2/2 \\
002 & Hard-rule composition and abstention & Balinese, one cited rule & 3/3 \\
003 & Travel note maps to \texttt{avoid\_needy} & No Oracle tag & 4/4 \\
004 & Tag penalty and base-score boundary & Balinese with correct tag & 5/5 \\
005 & Distinct soft rules compose & Incomplete soft ranking & 6/6 \\
006 & Duplicate soft concern applies once & Missing narrative tag & 7/7 \\
007 & Unknown allergy status escalates & Unsafe recommendation & 8/8 \\
008 & Malformed applicable policy escalates & Ignored policy error & 9/9 \\
\bottomrule
\end{tabular}
\end{center}

The separation between generations 003 and 004 is deliberate.  Generation 003 checks only the
prompt-mediated tag.  Its failure packet contains no breed values and leaves the deterministic
meaning of the tag open.  The next proposed case then fails because the correct tag does not yet
produce the approved ranking.  That independent failure becomes the fourth counterexample.

Every one of the eight accepted CatSynth counterexamples changes \texttt{SKETCH.md}.  CatSynth
also retains every accepted case as a regression because the run is small, so
$\regression = \archive$ in this example.  That equality is an experiment choice, not the
general method.

The main decision is not ``iteration or one shot.''  It is whether the governing policy is
already complete.  In a separate closed-world run, the model receives an immutable complete
specification and empty implementation files.  It reaches 20/20 visible and passes 21/21 withheld cases
after four Developer calls.  Reaching visible acceptance consumes 611,519 model tokens: 132,632
for Developer calls and 478,887 for prompt-mediated Runtime Oracle checks.  Post-acceptance
visible and withheld evaluation consumes another 239,929, for 851,448 recorded tokens in total.
When the closed-world premise is true, spec-first is the simpler approach.

A 2026-08-02 continuation reran final acceptance under the complete two-check protocol, with
manual approval of every sketch change and an explicit change-authority and preservation contract
in every Developer prompt.  Replay-all, evolved-sketch rebuild, and retained Sketch--CE all passed
8/8 visible accepted cases and then passed 14/21, 17/21, and 16/21 withheld cases, respectively.
This is the current withheld result: the evolved sketch retains a three-case advantage over raw
replay, while retained code does not beat clean regeneration from the reviewed sketch in this
sample.  Sketch review also rejected premature empty-input and tag policy, restored dropped
anchors, and kept adjudicated reviewer errors from becoming policy.

The earlier open-world capture adds two controls to isolate what carries learned policy.  Replay-all rebuilds
from the initial sketch and all accepted cases known at each epoch, asking the model to infer the
rules again from raw examples.  Evolved-sketch rebuild discards the implementation and receives
only the current evolved sketch and anchors; after a failed gate it receives one visible failure
packet at a time.  It is the clean-regeneration test built into the method.  Tokens through
visible acceptance are 1,021,822 for retained
Sketch--CE, 891,880 for replay-all, and 828,628 for evolved-sketch rebuild.  The Sketch--CE total
comprises 217,576 Developer tokens, 657,478 Runtime Oracle tokens, and 146,768 Specification
Oracle tokens.  Its candidate cases are external inputs: the extra cost comes from evaluating
them and proposing general rules for the failures, not from generating them.  The controls
inherit the resulting promotion schedule, so these totals have different accounting boundaries
and are not end-to-end price rankings.  Post-acceptance evaluation adds about 170,000 tokens to
each path.  Their withheld-case results are 18/21, 15/21, and 19/21 respectively.  Sketch--CE reduces
implementation work and churn in this run, but does not achieve the highest withheld-case
rate.  Its final strategy is also the largest and most branch-heavy: 298 lines and 110 decision
nodes, versus 224/77 for replay-all and 228/70 for evolved-sketch rebuild.  More importantly,
evolved-sketch rebuild passes 19/21 withheld cases versus 15/21 for replay-all.  That pattern
motivated the protocol-correct rerun, but the scores are not evidence for the current method
because this capture lacked required sketch review and approval.  The older table remains useful
for its captured cost and churn measurements, not as the headline evaluation of the two-check
method.

\section{Applying the method}
\label{sec:applying}

Start this loop in an ordinary repository.  Do not wait for a harness.  Use the files the team
already trusts: design notes, fixtures, parser tests, approved expected rows, prompt templates, or review
checklists.  For each accepted counterexample, record which operator approved it, which rule
changed, which example exposed the rule, and which regression would reject the tempting wrong
implementation.

\subsection{Start with the sketch}

Write the known strategy before asking the agent for code.  Put the first sketch in a short file that
names the allowed operations, priority order, known holes, and abstention cases.  Give later
failures a clause to attach to.  Do not present the first sketch as a complete policy.  Ask the
Developer to return the revised sketch with every CE-driven implementation generation, and
review that artifact with the code and prompt it describes.  Every accepted CE must change the
sketch.  If it does not, either the case was a regression rather than a counterexample or the
policy lesson has not been captured.

\subsection{Collect discriminating examples}

Keep examples that separate two plausible repairs.  Use a happy-path case to show the intended
output.  Use a proposed counterexample to show a repair that looks green but violates the rule.
When an SME corrects an agent, raise the case to an authorized operator if it counters the
sketch.  Do not approve examples only because they are available; approve the ones that teach a
missing or mistaken governing rule.

\subsection{Name known-code anchors}

Point the agent at the repository facts it should reuse: result types, parser behavior, fixture
format, error shape, prompt conventions, and naming rules.  Name the exact files, symbols, or
clauses when you can.  Use those anchors to keep the agent from inventing a second local
convention just to pass the current case.

\subsection{Separate replay from approved-output compare}

Build replay and approved-output compare as separate checks.  Use replay to check whether the repaired state
now satisfies the predicate that failed.  Use compare to check whether policy-bearing fields
still match the approved expectation.  Do not hide both jobs inside a single assertion that
says ``the output looks right.''  Read the failure by its job: replay failure means the state
is still broken; compare failure means the state repair passed while a policy field diverged.

\subsection{Review simulated outputs against the sketch}

After every repair, simulate the active case and $\regression$.  Give each input, output,
relevant trace, and the current sketch to a capable model or person.  Record the verdict and the
clauses that control it.  Do not reveal another candidate until this review and the deterministic
gate pass.

Judge the output before changing the sketch.  If the sketch already states the right behavior,
repair the implementation.  If the failure exposes missing or mistaken policy, raise a proposed
counterexample.  The same reviewer may make both decisions when authorized; keeping the
decisions separate prevents the implementation from excusing itself through a sketch edit.

\subsection{Approve counterexamples deliberately}

Triage each proposed case before adding it to $\archive$.  Run it first.  A passing proposal is
coverage, not a counterexample.  Treat fixture bugs as fixture work, checker bugs as checker work,
and behavior already governed by the sketch as an implementation regression.  Raise only
failures that expose missing or mistaken policy.  The operator must explicitly approve the case
and corrected behavior before Developer receives that one CE and revises $\sketch$ with the
implementation.

For each accepted case, record its approval, the tempting implementation it rejects, and the
sketch clause that changed.  Then choose whether that CE belongs in $\regression$ or whether an
already-selected CE protects the same boundary.  This work turns one corrected answer into an
evolved governing model without confusing the complete history with the executable subset.

\subsection{Prove that the sketch carries the learning}

Periodically discard generated code and prompts.  Give Developer only the current evolved
sketch and known-code anchors, then run $\gate(\regression)$ and review the simulated outputs
against the sketch.  Return one failed case at a time until both checks pass.  If regeneration
requires replaying the complete archive as prompt context, revise the sketch: it has not yet
synthesized the accepted policy well enough.

\section{Discussion}
\label{sec:discussion}

\subsection{Tests name behavior; sketches name intent}

A passing test suite covers only encoded behavior.  Tests check behavior in the current run;
sketch clauses preserve why a failure mattered for the next repair.  When a case enters
$\archive$, record the rule a future implementation must preserve alongside the input, approved
output, and operator decision.  Put a discriminating subset in $\regression$.

\subsection{Prompts implement one surface of the contract}

Treat prompts as code that implements policy.  If a rule lives in instructions or narrative
notes, bind that surface to the same sketch and regression gate as source code.  Review prompt
edits as contract changes.  Check their outputs at the deterministic gate and against the sketch
instead of treating them as informal wording.

\subsection{SME corrections become specification inputs}

Treat SME review as specification input.  When an SME corrects an output, ask for the rule the
output violated and raise it to the operator when it counters the sketch.  A correction becomes
specification only after explicit approval adds it to $\archive$ and causes a reviewed sketch
revision; until then it is feedback.

\subsection{Approval and regression selection are separate constraints}

Approval quality controls the governing model; regression selection controls executable scale.
Accept a CE when it exposes missing or mistaken sketch policy.  Preserve every accepted case in
$\archive$.  Retain in $\regression$ only the cases needed to reject the
known wrong implementations and protect distinct boundaries.

\section{Limitations}
\label{sec:limitations}

\begin{enumerate}
  \item \textbf{Finite regression set.}  A passing gate proves only the current
        $\regression$.  The archive may contain accepted cases not selected for routine
        execution, and the result does not extend to unseen cases.
  \item \textbf{Checker quality.}  Replay and approved-output compare are only as good as the
        semantics encoded in their checkers.  A buggy checker can certify the wrong
        behavior until a maintainer repairs the checker and reruns the regression set.
  \item \textbf{Golden dataset quality.}  A wrong row in the golden dataset makes
        the gate enforce the wrong behavior.  A domain reviewer must distinguish a
        corrected answer from a corrected specification before the row becomes part
        of the claim.
  \item \textbf{Approval quality.}  An operator should approve a counterexample only when it
        exposes missing or mistaken sketch policy and names authoritative corrected behavior.
        Developer must not authorize its own policy changes.
  \item \textbf{Regression selection.}  A small $\regression$ can omit an important boundary;
        an indiscriminate one can become expensive noise.  Maintainers must preserve the full
        archive while selecting checks that reject distinct known wrong implementations.
  \item \textbf{Sketch quality.}  Rules left outside the sketch still rely on human memory.
        Every accepted CE must revise the sketch, and clean regeneration should test whether
        those revisions are sufficient.
  \item \textbf{Sketch-review quality.}  A model or person can misread the sketch, share the
        implementation's blind spots, or return inconsistent judgments.  Record the reviewer,
        inputs, outputs, applicable clauses, and verdict.  Sketch review adds evidence; it is not
        proof.
  \item \textbf{Prompt drift.}  Live model behavior can drift.  Review live output
        before approval; do not treat a prompt response as a proof artifact.
  \item \textbf{Captured-protocol boundary.}  The published CatSynth run records the
        deterministic gate after each repair but not a separate post-repair sketch review over
        the active case and selected regressions.  Its results do not measure that review step or
        its cost.
  \item \textbf{Single captured comparisons.}  CatSynth records one stochastic run per path with
        one model and candidate order.  The closed-world run begins with information the
        open-world run must discover, so their token totals answer different questions.  Within
        the open-world experiment, the rebuild controls inherit the promoted stream without
        paying to evaluate the external candidates or propose rules for failures.  The observed
        token ratios and checked results are not a
        benchmark or a causal estimate.
  \item \textbf{Enterprise evidence boundary.}  Readers can inspect the control
        structure in CatSynth's synthetic fixtures.  The reported evidence covers that control
        structure, not private deployment behavior.
\end{enumerate}

These boundaries mark where the proof stops.  The gate checks the finite regression set with
current replay, approved-output compare, and data.  Sketch review remains a recorded judgment.
People still own counterexample approval, regression selection, checker design, golden review,
sketch maintenance, prompt review, and any private deployment claim.

\section{Conclusion}
\label{sec:conclusion}

Do not trust an agent repair just because the patch looks plausible.  When a failure counters
the sketch, require an SME correction and explicit operator approval.  Then change the sketch,
not just the code.

For an accepted case, the required trail is: the SME correction, operator approval, archive
record, revised sketch, one projected failure given to Developer, repaired or regenerated
implementation, the selected regression that protects the learned boundary, and review of the
simulated active and regression outputs against the current sketch.

When those pieces are present, the agent has a bounded job: revise the sketch, code, and prompt
for one approved counterexample, or repair one regression under the current sketch.  The sketch
carries the learned policy.  The archive carries the decision history.  The regression set tests
generated implementations.  The deterministic gate checks encoded behavior; sketch review
checks the simulated output against the current governing model.  The code can be discarded and
regenerated.

A passing gate does not prove the domain correct.  It says only this: for the current repository,
current checkers, and current regression set, replay and approved-output comparison pass.  Everything
outside those encoded cases and predicates still belongs to review and future discovery.  The
workflow therefore requires a capable model or person to review the active case and selected
regressions against the sketch before another case is revealed.

\section*{Artifact availability}

The argument, method, experimental design, reported results, and limitations are complete in
this paper.  Appendix~\ref{app:catsynth-supplement} provides the illustrated audit and
reproduction companion.  The public artifact adds CatSynth's synthetic fixtures, initial and
evolved sketches, Developer and Oracle adapters, generated code and prompts, per-generation gate
outcomes and diffs, compact token ledgers, browser application, and full experiment histories.
The artifact is available in the
\href{https://github.com/open-horizon-labs/counterexample-supplemented-sketches}{public repository}.

\section*{Use of generative AI}

The CatSynth experiments use generative AI as described in
Section~\ref{sec:worked-example} and Appendix~\ref{app:catsynth-supplement}.  OpenAI Codex also
assisted with experiment orchestration, code and source editing, and prose revision.  The
authors reviewed the generated code, prompts, analyses, citations, and text and take
responsibility for this paper.

\bibliographystyle{plain}
\small
\bibliography{references}

\normalsize
\clearpage
\appendix
\section{CatSynth Artifact Supplement}
\label{app:catsynth-supplement}

\begingroup
\let\section\subsection
\let\subsection\subsubsection
\input{catsynth-worked-example}
\endgroup

\end{document}

%% file: pandoc-support.tex
\usepackage{lmodern}
\usepackage{textcomp}
\usepackage{longtable}
\usepackage{array}
\usepackage{calc}
\usepackage{etoolbox}

\newcommand{\paperlicense}{%
  \begin{center}
  \small
  \textcopyright\ 2026 Muness Castle and Eric Rubeck. Except where otherwise noted, this
  paper and its original figures are licensed under
  \href{https://creativecommons.org/licenses/by/4.0/}{Creative Commons Attribution 4.0
  International (CC BY 4.0)}. System-rendered emoji glyphs shown in CatSynth screenshots are
  excluded from this license.
  \end{center}
}

\makeatletter
\patchcmd\longtable{\par}{\if@noskipsec\mbox{}\fi\par}{}{}

\newsavebox\pandoc@box
\newcommand*\pandocbounded[1]{%
  \noindent
  \sbox\pandoc@box{#1}%
  \Gscale@div\@tempa{\textheight}{\dimexpr\ht\pandoc@box+\dp\pandoc@box\relax}%
  \Gscale@div\@tempb{\linewidth}{\wd\pandoc@box}%
  \ifdim\@tempb\p@<\@tempa\p@\let\@tempa\@tempb\fi
  \makebox[\linewidth][c]{%
    \ifdim\@tempa\p@<\p@\scalebox{\@tempa}{\usebox\pandoc@box}%
    \else\usebox{\pandoc@box}%
    \fi
  }%
}
\makeatother

\providecommand{\tightlist}{%
  \setlength{\itemsep}{0pt}\setlength{\parskip}{0pt}}

%% file: catsynth-worked-example.tex
This is the audit and reproduction companion to \emph{Agentic Synthesis against
Counterexample-Supplemented Sketches}. The paper contains the complete argument, method,
experimental design, reported results, and limitations. The distributed paper appends this same
material; \texttt{catsynth-supplement.pdf} packages it separately for readers who want the implementation
trace on its own. Nothing in the supplement extends the paper\textquotesingle s method or claims.

This supplement provides the implementation depth that would interrupt the paper: the complete
CatSynth generation sequence, screenshots, reproduction commands, source map, and links from
each counterexample to the revised sketch, code, prompt, and gate result. The browser makes the
finished artifacts visible. The experiment harness asks a coding model to evolve the sketch,
deterministic code, and model prompt one counterexample at a time.

CatSynth uses synthetic breed attributes and policy rows selected to make the control loop easy
to inspect. Nothing here is pet-selection or medical advice.

\section{The artifact boundary}\label{the-artifact-boundary}

For orientation, CatSynth maps the paper\textquotesingle s method onto four operational responsibilities:

\begin{itemize}
\tightlist
\item
  The \textbf{operator} decides whether a failing case is an authoritative counterexample to the
  current sketch.
\item
  The \textbf{evolved sketch} carries the policy learned from every accepted counterexample.
\item
  The \textbf{CE archive} preserves every approved case and explains why the sketch changed.
\item
  The \textbf{regression set} checks selected consequences against generated code.
\end{itemize}

The code and prompt are replaceable. A clean implementation should be regenerable from the
evolved sketch and known-code anchors without replaying the CE archive as generation context.

The current method requires two checks before revealing another candidate: the deterministic
gate runs the active case and \texttt{R}, and a capable model or person reviews those simulated outputs
against the current sketch. The same reviewer may also approve a sketch change when authorized,
but judging the output and changing the sketch remain separate decisions.

CatSynth makes one simplifying choice: its run is small, so every accepted CE is also a
regression case (\texttt{R\ =\ A}). That is why its gate sees every promoted case. The general method does
not require the full archive to remain in the executable regression set.

The iterative arm also obeys one generation boundary:

\begin{quote}
Developer sees one active failure. It never receives the CE archive or regression set as a
bulk prompt.
\end{quote}

The run starts from an initial sketch and empty \texttt{strategy.py} and \texttt{oracle\_prompt.txt} files.
Developer returns complete replacements for all three evolving artifacts:

\begin{verbatim}
SKETCH.md
strategy.py
oracle_prompt.txt
\end{verbatim}

After the initial implementation passes its anchor, the harness reveals one proposed
counterexample. It evaluates the case before approval. A case that already passes is coverage,
not a counterexample; the harness records that result and continues to the next proposal without
sending it to Developer.

A failing case becomes a CE only if it exposes missing sketch policy and the operator approves
the corrected behavior. Developer then receives the current three files and that one failure. It
must revise the sketch with the code or prompt. The gate runs the active case and current
regressions. If an earlier case regresses, that failed regression becomes the next single
Developer input. The harness reveals no new case until the gate is green.

That sentence describes the captured experiment. The capture predates the explicit requirement
to record a separate post-repair sketch review over the active case and \texttt{R}. It therefore tests
sketch evolution and the deterministic gate, not the complete two-check acceptance rule or the
cost of repeated sketch review.

On August 2, 2026, a corrected two-check continuation completed the GPT-5.4-mini benchmark after
the initial wrapper had stopped at an arbitrary 24-repair bound. Its final visible and withheld
scores are reported below. GPT-5.3-codex-spark remains incomplete at CE-012 and is a bounded
failure record, not a model ranking. The compact record is under
\texttt{examples/catsynth/experiment/results/two-check-reruns-20260802/}.

The captured run freezes candidate cases and authoritative expected outputs before execution,
then treats them as a simulated operator-approved stream. That makes the experiment
reproducible, but it simulates the live approval step rather than giving the model authority to
approve policy. Every accepted CE in the captured run changes \texttt{SKETCH.md}. The informational
restriction prevents Developer from copying future counterexamples into the sketch because it
never sees them.

\section{Reproduce the run}\label{reproduce-the-run}

From \texttt{examples/catsynth}:

\begin{verbatim}
uv run --with-requirements requirements.txt \
  python experiment/adaptive_open_world_experiment.py \
  --model gpt-5.4-mini \
  --max-repairs 12
\end{verbatim}

This published three-path driver uses Codex App Server. The run used GPT-5.4-mini with low effort.
The adapter disabled tools and environment access, disabled provider fallback, and used an
ephemeral thread for every call. It consumed 173 model calls and 3,251,645 recorded model tokens
across all three paths, including post-acceptance evaluation. Raw JSON-RPC transcripts stayed
local; the repository retains every generated sketch, strategy, prompt, failure, gate outcome,
diff, and usage total. The \href{https://github.com/open-horizon-labs/counterexample-supplemented-sketches/tree/main/examples/catsynth}{CatSynth README}
documents the portable driver for Codex App Server and OpenAI-compatible Chat Completions
endpoints.

The captured evidence is in the
\href{https://github.com/open-horizon-labs/counterexample-supplemented-sketches/tree/main/examples/catsynth/experiment/results/gpt-5.4-mini-adaptive-open-world-v2-20260712}{checked-in run directory}.

\section{The starting point}\label{the-starting-point}

The initial sketch fixes the public interface and the known preference ranking:

\begin{verbatim}
recommend(profile, breeds, rules, oracle_tags)
\end{verbatim}

It defines the ordinal maps and exact preference weights, but deliberately leaves two policy
surfaces open:

\begin{itemize}
\tightlist
\item
  rule rows exist, but the sketch does not yet say what matched rules do;
\item
  \texttt{oracle\_tags} exists, but no tag has a meaning yet.
\end{itemize}

That is the partial specification. It contains enough detail to generate and test an initial
strategy without preloading the later policy discoveries.

The clean-room baseline contains no implementation. Both rebuild controls copy the same empty
files at each discovery epoch.

\section{Generation 000: the initial strategy}\label{generation-000-the-initial-strategy}

Developer generates the first complete sketch, deterministic strategy, and narrative prompt
from the partial sketch and empty implementation files. The initial preference anchor passes:

\begin{verbatim}
initial-preference-ranking: PASS
gate: 1/1
\end{verbatim}

Only then does the harness reveal the first domain counterexample. The complete generation is
preserved under \texttt{arms/sketch-ce/generations/000-initial-generation/}.

\section{Generation 001: hard policy must filter before ranking}\label{generation-001-hard-policy-must-filter-before-ranking}

The first profile describes an owner with mild allergies who wants a large, fluffy,
affectionate cat. The current preference-only strategy returns Persian.

The approved counterexample requires:

\begin{verbatim}
operation:    recommend
breed:        Siberian
cited_rules:  [allergy_requires_hypoallergenic]
\end{verbatim}

The synthetic policy row says that mild or severe allergies activate a hard \texttt{forbid} rule for
breeds whose \texttt{hypoallergenic} field is false. The correction adds a general ordering rule:
filter hard-policy violations before ranking the survivors.

Developer receives only this failure and the current files. It revises the sketch and
\texttt{strategy.py} to interpret the rule row generically. CatSynth\textquotesingle s \texttt{R\ =\ A} gate then passes the initial anchor
and CE1:

\begin{verbatim}
initial-preference-ranking: PASS
ce-001-allergy-override:    PASS
gate: 2/2
\end{verbatim}

The browser presents the same near miss as a teaching surface:

\pandocbounded{\includegraphics[keepaspectratio,alt={Naive resolver choosing the tempting Persian output}]{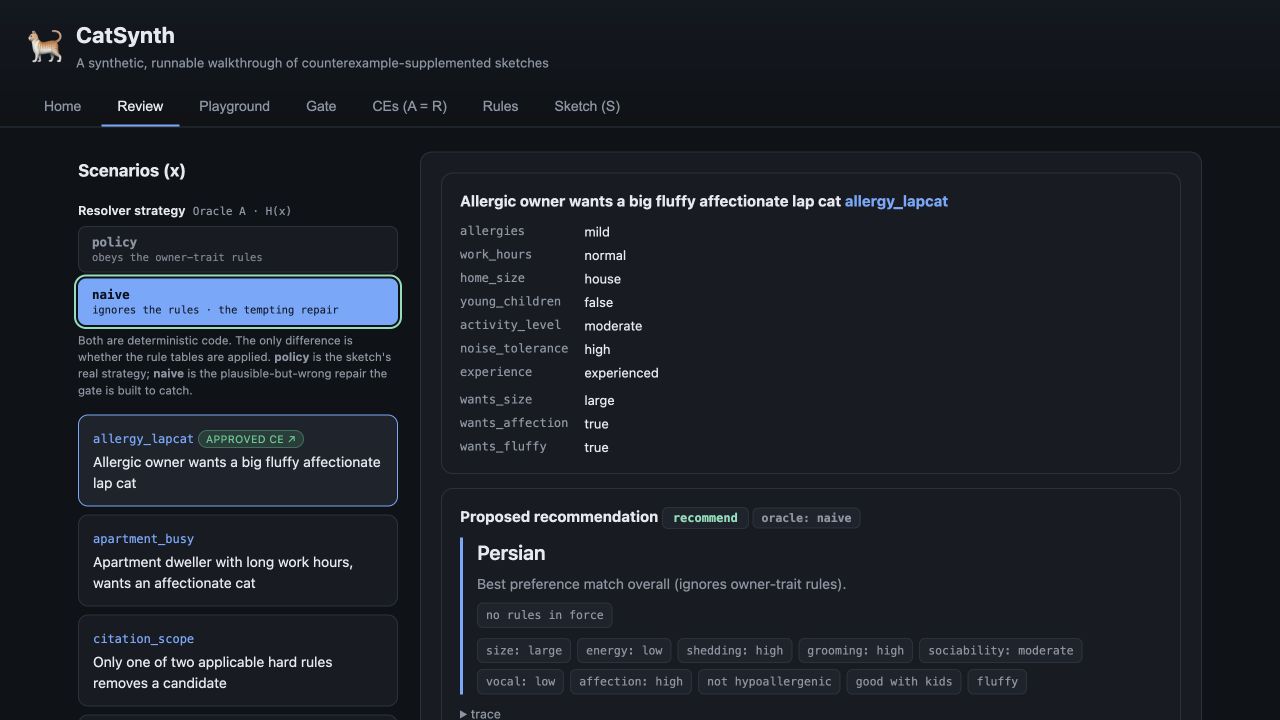}}

The point is the encoded ordering, not the cat claim. Persian closes the visible preference gap;
Siberian also closes it while satisfying the approved hard rule.

\section{Generation 002: hard rules compose and may force abstention}\label{generation-002-hard-rules-compose-and-may-force-abstention}

The second counterexample activates five hard rules: severe allergies, apartment constraints,
long work hours, and young children. The current implementation understands the first allergy
operator but still returns Balinese and cites only that rule.

The approved expectation is:

\begin{verbatim}
operation: abstain
breed: null
cited_rules:
  - allergy_requires_hypoallergenic
  - apartment_no_high_energy
  - children_require_good_with_children
  - long_hours_no_high_sociability
  - severe_allergy_low_shedding
\end{verbatim}

Developer generalizes again. The revised sketch and code support the additional profile and cat
predicate operators, apply every applicable hard rule, and abstain rather than relaxing policy
when no candidate remains.

The full regression passes:

\begin{verbatim}
initial anchor: PASS
CE1:            PASS
CE2:            PASS
gate:           3/3
\end{verbatim}

\section{Generation 003: the prompt learns a controlled tag}\label{generation-003-the-prompt-learns-a-controlled-tag}

The third profile has no new structured policy row. Its missing information is in the narrative
note:

\begin{verbatim}
I travel for work every few weeks and my last cat seemed miserable and lonely
whenever I was gone for days.
\end{verbatim}

Before approval, the current prompt emits no tags and the deterministic strategy recommends
Balinese. CE3 adds one controlled narrative classification to the evolved sketch:

\begin{itemize}
\tightlist
\item
  travel, repeated absence, or concern about loneliness maps to exactly \texttt{avoid\_needy};
\item
  the prompt classifies only the supplied note and may not invent unrelated tags;
\item
  the deterministic meaning of \texttt{avoid\_needy} remains an open hole.
\end{itemize}

This case compares only \texttt{oracle\_tags}. Developer revises the prompt and sketch. The gate
passes 4/4 even though the strategy still chooses Balinese, because CE3 has not yet added a
deterministic meaning for the tag.

\begin{verbatim}
initial anchor: PASS
CE1:            PASS
CE2:            PASS
CE3 tags:       PASS
gate:           4/4
\end{verbatim}

\section{Generation 004: a new counterexample closes the deterministic hole}\label{generation-004-a-new-counterexample-closes-the-deterministic-hole}

After CE3, the prompt emits \texttt{avoid\_needy}, but the strategy still recommends Balinese. The
harness evaluates the next proposed case before approval. It fails on the breed field, so CE4
is a genuine new counterexample rather than another explanation pasted into CE3.

CE4 adds two connected rules:

\begin{itemize}
\tightlist
\item
  \texttt{avoid\_needy} applies a one-point soft penalty to breeds with high sociability;
\item
  base preference scoring continues to use only the three explicit \texttt{wants\_*} fields: size,
  affection, and fluffiness. Default activity, noise, and experience values do not add score
  unless an approved sketch clause gives them semantics.
\end{itemize}

Developer revises the sketch and deterministic code. The prompt remains green. The gate
passes 5/5:

\begin{verbatim}
initial anchor: PASS
CE1:            PASS
CE2:            PASS
CE3 tags:       PASS
CE4 ranking:    PASS
gate:           5/5
\end{verbatim}

\section{Generation 005: distinct soft rules compose}\label{generation-005-distinct-soft-rules-compose}

The next accepted case activates three different soft policy predicates. The current strategy
returns Abyssinian because it stops short of applying the full soft-policy total. The approved
output is British Shorthair.

CE6 adds a general rule: every distinct applicable \texttt{discourage} predicate contributes before the
final ranking and tie-break. Developer revises the sketch and strategy; the complete gate passes
6/6.

\section{Generation 006: duplicate concerns count once across surfaces}\label{generation-006-duplicate-concerns-count-once-across-surfaces}

The next profile expresses the same high-energy concern twice: once through a structured policy
row and once through a narrative note. Before repair, the prompt emits no tag. The approved output
requires \texttt{avoid\_high\_energy}, while the ranking must apply the shared energy predicate only once.

CE7 makes the sketch join structured rules and narrative tags by the semantic triple of cat
attribute, operator, and value. Developer revises the sketch, prompt, and strategy; the complete
gate passes 7/7.

\section{Generation 007: unknown safety data escalates}\label{generation-007-unknown-safety-data-escalates}

The next profile supplies \texttt{unknown} for allergy status. The current strategy treats it like no
allergy and recommends Persian. CE10 distinguishes missing or unsupported safety data from a
negative value: the implementation must return \texttt{escalate} with no breed until an operator can
clarify the input.

Developer records the rule in the sketch and implements the escalation. The complete gate passes
8/8.

\section{Generation 008: malformed applicable policy escalates with provenance}\label{generation-008-malformed-applicable-policy-escalates-with-provenance}

The final accepted case supplies an applicable hard policy row whose cat operator is unsupported.
The current strategy silently ignores it and recommends Persian. CE12 requires \texttt{escalate} and
cites \texttt{invalid\_reviewer\_policy}, preserving which policy source needs repair.

Developer adds the validation and provenance rule to the sketch and strategy. The full retained
gate passes the initial anchor and all eight accepted counterexamples: 9/9.

In the historical deterministic-only capture, the final retained Sketch-CE implementation passes 18/21 withheld cases. It misses two multi-tag
cases because no accepted discovery has yet defined \texttt{avoid\_vocal}, and it misses one normalized
severe-allergy variant. Those failures show where another open-world counterexample could extend
the current sketch.

\section{What each generation archive proves}\label{what-each-generation-archive-proves}

Every generation directory under \texttt{arms/} contains:

{\def\LTcaptype{none} 
\begin{longtable}[]{@{}
  >{\raggedright\arraybackslash}p{(\linewidth - 2\tabcolsep) * \real{0.3100}}
  >{\raggedright\arraybackslash}p{(\linewidth - 2\tabcolsep) * \real{0.6900}}@{}}
\toprule\noalign{}
\begin{minipage}[b]{\linewidth}\raggedright
File
\end{minipage} & \begin{minipage}[b]{\linewidth}\raggedright
Evidence
\end{minipage} \\
\midrule\noalign{}
\endhead
\bottomrule\noalign{}
\endlastfoot
\texttt{SKETCH.md} & Developer\textquotesingle s complete revised strategy \\
\texttt{strategy.py} & Complete deterministic implementation \\
\texttt{oracle\_prompt.txt} & Complete prompt implementation \\
\texttt{metadata.json} & Active failure or failures, accepted CE IDs, compact gate outcome, usage, and diffs \\
\end{longtable}
}

The Sketch-CE repair metadata identifies one active failure and no unrevealed case. The rebuild
controls preserve each generated state and the visible failures returned after a failed gate.
A reader can therefore inspect the information boundary and code evolution directly.

\section{Audit questions}\label{audit-questions}

The retained histories let a reader answer five questions without relying on the paper\textquotesingle s prose:

\begin{enumerate}
\def\labelenumi{\arabic{enumi}.}
\tightlist
\item
  Which single failure was visible to each Developer generation?
\item
  How did the sketch, deterministic code, and prompt change together?
\item
  How did every accepted CE change the sketch?
\item
  Which regression gate ran after each revision?
\item
  Which proposed cases became accepted CEs, and which were recorded only as coverage?
\end{enumerate}

\section{The open-world comparison}\label{the-open-world-comparison}

The conceptual contrast remains spec-first versus Sketch-CE: a complete specification works when
the problem is already known, while Sketch-CE changes the governing sketch as the world reveals
new policy. The captured open-world experiment adds two controls to identify what carries that
policy forward.

\begin{itemize}
\tightlist
\item
  \textbf{Replay-all} rebuilds from the initial sketch and every accepted case known at the current
  discovery epoch. It asks the model to infer policy again from raw examples.
\item
  \textbf{Evolved-sketch rebuild} discards code and prompt, then receives only the current evolved
  sketch and known-code anchors. If its gate fails, it receives one visible failure at a time.
  This is the method\textquotesingle s clean-regeneration test.
\item
  \textbf{Sketch-CE with retained code} keeps the current sketch, code, and prompt and repairs each
  newly accepted case.
\end{itemize}

{\def\LTcaptype{none} 
\begin{longtable}[]{@{}
  >{\raggedright\arraybackslash}p{(\linewidth - 6\tabcolsep) * \real{0.4000}}
  >{\raggedleft\arraybackslash}p{(\linewidth - 6\tabcolsep) * \real{0.2000}}
  >{\raggedleft\arraybackslash}p{(\linewidth - 6\tabcolsep) * \real{0.2000}}
  >{\raggedleft\arraybackslash}p{(\linewidth - 6\tabcolsep) * \real{0.2000}}@{}}
\toprule\noalign{}
\begin{minipage}[b]{\linewidth}\raggedright
Measure
\end{minipage} & \begin{minipage}[b]{\linewidth}\raggedleft
Replay-all
\end{minipage} & \begin{minipage}[b]{\linewidth}\raggedleft
Evolved-sketch rebuild
\end{minipage} & \begin{minipage}[b]{\linewidth}\raggedleft
Sketch-CE (retained code)
\end{minipage} \\
\midrule\noalign{}
\endhead
\bottomrule\noalign{}
\endlastfoot
\textbf{All recorded model tokens, including evaluation} & \textbf{1,061,834} & \textbf{998,307} & \textbf{1,191,504} \\
Tokens through visible acceptance & 891,880 & 828,628 & 1,021,822 \\
Post-acceptance evaluation tokens & 169,954 & 169,679 & 169,682 \\
Developer calls & 15 & 16 & 9 \\
Developer tokens & 400,081 & 371,050 & 217,576 \\
Runtime Oracle tokens through acceptance & 491,799 & 457,578 & 657,478 \\
Specification Oracle tokens & 0 & 0 & 146,768 \\
Rebuilds & 9 & 9 & 1 \\
Extra repair attempts & 6 & 7 & 0 \\
Prior regressions on first attempt & 2 & 7 & 0 \\
Artifact churn lines & 2,394 & 2,326 & 719 \\
Final strategy LOC & 224 & 228 & 298 \\
Final decision nodes & 77 & 70 & 110 \\
Accepted CE evaluation & 8/8 & 8/8 & 8/8 \\
Withheld evaluation & 15/21 & 19/21 & 18/21 \\
\end{longtable}
}

The first row includes every recorded Developer, Runtime Oracle, Specification Oracle, and
post-acceptance evaluation call. Tokens through acceptance exclude only the final visible and
withheld evaluation. Provider totals count input plus output; cached input and reasoning are
subsets, not additional tokens.

The candidate cases were external inputs. Sketch-CE paid to classify them and propose general
rules for failures. Both controls inherited the promotion schedule, and evolved-sketch rebuild
also inherited the sketch checkpoints. Their totals omit discovery work and are not end-to-end
price rankings.

Sketch-CE with retained code used less Developer work and produced less cumulative churn.
Evolved-sketch rebuild passed 19/21 withheld cases versus 15/21 for Replay-all. That pattern
motivated the protocol-correct continuation below. Because this capture lacked required sketch
review and approval, the scores are not evidence for the current two-check method. The retained
strategy was the largest and had the most decision nodes, so that path shows less rework during
evolution, not better final maintainability.

This is one model, one candidate order, and one sample per path. It does not establish universal
cost, correctness, or maintainability superiority.

\section{Protocol-correct continuation}\label{protocol-correct-continuation}

On 2026-08-02, the same mini arm was completed with deterministic acceptance, separate review of
each simulated output against the current sketch, manual approval of every sketch change, and an
explicit authority-and-preservation contract in every Developer prompt.

{\def\LTcaptype{none} 
\begin{longtable}[]{@{}
  >{\raggedright\arraybackslash}p{(\linewidth - 6\tabcolsep) * \real{0.4000}}
  >{\raggedleft\arraybackslash}p{(\linewidth - 6\tabcolsep) * \real{0.2000}}
  >{\raggedleft\arraybackslash}p{(\linewidth - 6\tabcolsep) * \real{0.2000}}
  >{\raggedleft\arraybackslash}p{(\linewidth - 6\tabcolsep) * \real{0.2000}}@{}}
\toprule\noalign{}
\begin{minipage}[b]{\linewidth}\raggedright
Evaluation
\end{minipage} & \begin{minipage}[b]{\linewidth}\raggedleft
Replay-all
\end{minipage} & \begin{minipage}[b]{\linewidth}\raggedleft
Evolved-sketch rebuild
\end{minipage} & \begin{minipage}[b]{\linewidth}\raggedleft
Sketch-CE (retained code)
\end{minipage} \\
\midrule\noalign{}
\endhead
\bottomrule\noalign{}
\endlastfoot
Visible accepted cases & 8/8 & 8/8 & 8/8 \\
Withheld cases & 14/21 & 17/21 & 16/21 \\
\end{longtable}
}

This replaces the old 19/21-versus-15/21 headline for the current method. The direction remains:
the reviewed evolved sketch beats raw replay, now by three withheld cases. Retaining code finishes
between the controls and does not beat clean regeneration from the sketch in this sample. The old
table remains the captured source for token, repair, and churn measurements.

The clarification changed what the experiment can observe. The old capture could record a passing
repair but not whether its draft sketch filled an unresolved hole or dropped an existing anchor.
The two-check run rejected drafts that assigned an empty-catalog operation before CE11, gave the
CE3 \texttt{avoid\_needy} tag a ranking effect before CE4, or removed stable input-shape clauses. It also
adjudicated reviewer math errors rather than converting them into policy changes. The result is
not an alignment claim: it makes proposed policy drift reviewable before it becomes accepted
sketch text.

\section{How the UI illustrates the finished artifacts}\label{how-the-ui-illustrates-the-finished-artifacts}

The experiment is the synthesis history. The browser is the finished inspection surface.

\begin{verbatim}
cd examples/catsynth
uv run --with-requirements requirements.txt python cli.py seed --no-wiki
uv run --with-requirements requirements.txt python cli.py serve
\end{verbatim}

Open \url{http://127.0.0.1:8000}.

\pandocbounded{\includegraphics[keepaspectratio,alt={CatSynth home screen mapping the method artifacts}]{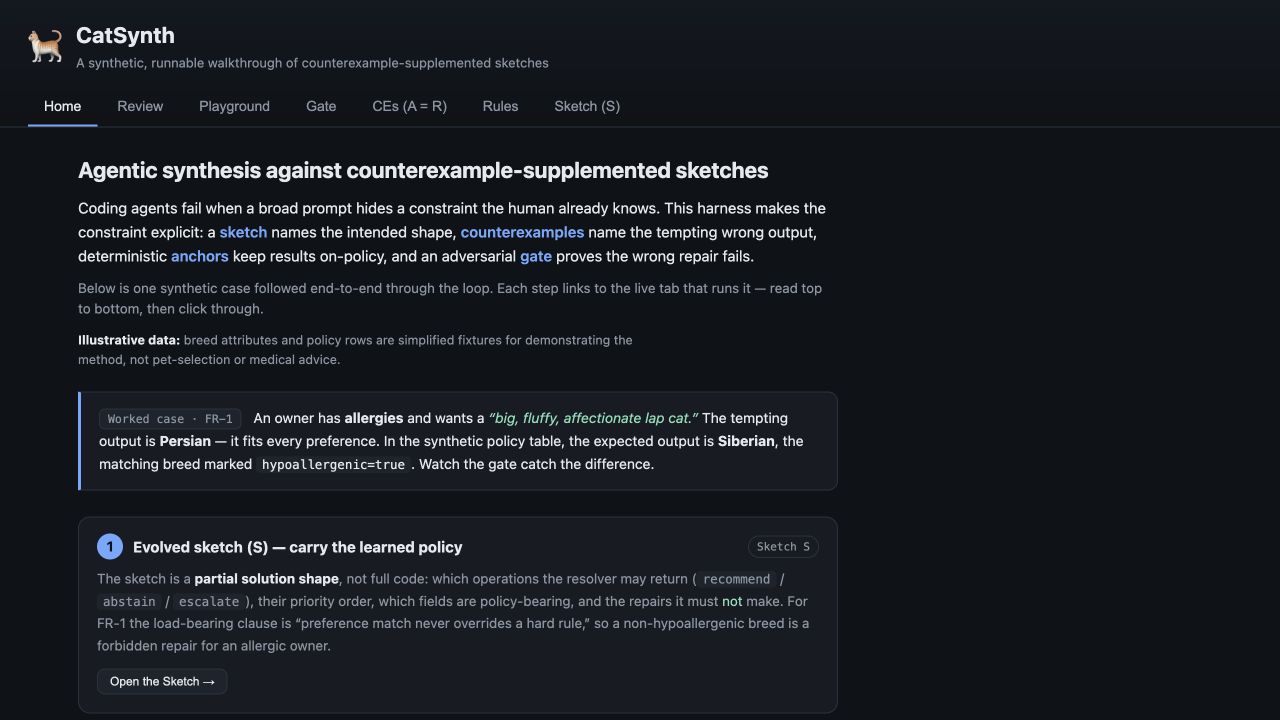}}

The UI exposes the stable repository artifacts:

\begin{enumerate}
\def\labelenumi{\arabic{enumi}.}
\tightlist
\item
  \textbf{Sketch} - the strategy, policy order, holes, and abstention behavior.
\item
  \textbf{CE archive / regression set} - approved expected outputs, tempting outputs, violated rules,
  and sketch links. CatSynth shows the same cases in both roles because \texttt{R\ =\ A}.
\item
  \textbf{Oracle A} - deterministic hard-rule filtering, ranking, and abstention.
\item
  \textbf{Oracle B} - prompt-mediated narrative tags constrained to a controlled vocabulary.
\item
  \textbf{Gate} - replay and approved-output comparison over CatSynth\textquotesingle s regression set.
\end{enumerate}

The current method also requires sketch review over the active case and \texttt{R}. The browser lets a
person inspect the sketch and concrete outputs, but the captured experiment does not record that
review after every repair.

The browser\textquotesingle s Review button records only a local operator decision in SQLite. It does not invoke
Developer or revise \texttt{SKETCH.md}. The experiment history, not that button, is the evidence for the
complete CE-to-sketch-and-code loop.

The CE archive / regression page keeps both sides of the focal correction:

\pandocbounded{\includegraphics[keepaspectratio,alt={Approved CatSynth counterexample in archive A and regression set R}]{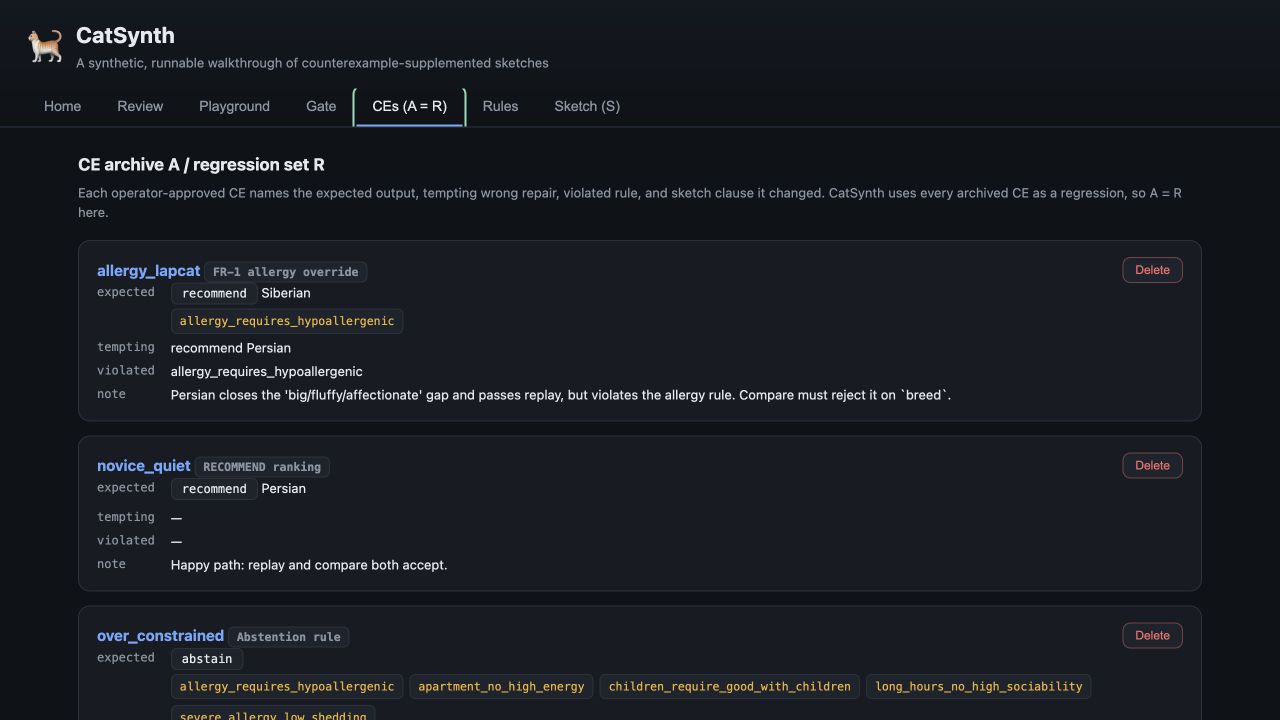}}

An expected output records what should happen. The tempting output and violated rule record why a
plausible alternative must continue to fail.

\section{Why replay and approved-output compare stay separate}\label{why-replay-and-approved-output-compare-stay-separate}

Replay asks whether the candidate closes the visible state gap. For the focal recommendation, it
checks the encoded size, affection, and fluffiness preferences. It deliberately does not decide
the hard allergy policy.

Approved-output compare checks the approved policy-bearing fields:

\begin{verbatim}
operation
breed
cited_rules
\end{verbatim}

The naive Persian result can therefore pass replay and fail approved-output compare:

\includegraphics[width=0.85\linewidth,height=\textheight,keepaspectratio,alt={Naive gate where replay passes and approved-output compare fails}]{figures/catsynth/04-naive-gate.png}

The split makes the failure actionable. The candidate repaired the visible preference state but
did not follow the approved policy.

Sketch review is a third question: does the simulated output follow the current sketch, including
meaning the deterministic checkers do not encode? A capable model or person performs that review.
The captured screenshot still labels approved-output compare as ``semantic compare,'' the original
terminology.

\section{Source map}\label{source-map}

{\def\LTcaptype{none} 
\begin{longtable}[]{@{}
  >{\raggedright\arraybackslash}p{(\linewidth - 2\tabcolsep) * \real{0.3100}}
  >{\raggedright\arraybackslash}p{(\linewidth - 2\tabcolsep) * \real{0.6900}}@{}}
\toprule\noalign{}
\begin{minipage}[b]{\linewidth}\raggedright
Paper concept
\end{minipage} & \begin{minipage}[b]{\linewidth}\raggedright
Repository artifact
\end{minipage} \\
\midrule\noalign{}
\endhead
\bottomrule\noalign{}
\endlastfoot
Initial sketch & \href{https://github.com/open-horizon-labs/counterexample-supplemented-sketches/blob/main/examples/catsynth/experiment/initial_sketch.md}{\texttt{initial\_sketch.md}} \\
Candidate manifest & \href{https://github.com/open-horizon-labs/counterexample-supplemented-sketches/blob/main/examples/catsynth/experiment/adaptive_candidate_manifest.json}{\texttt{adaptive\_candidate\_manifest.json}} \\
Operator references & \href{https://github.com/open-horizon-labs/counterexample-supplemented-sketches/blob/main/examples/catsynth/experiment/cases.json}{\texttt{cases.json}} \\
Accepted CE archive and regression set (\texttt{R\ =\ A}) & \href{https://github.com/open-horizon-labs/counterexample-supplemented-sketches/blob/main/examples/catsynth/experiment/results/gpt-5.4-mini-adaptive-open-world-v2-20260712/promoted-corpus.json}{\texttt{promoted-corpus.json}} \\
Experiment driver & \href{https://github.com/open-horizon-labs/counterexample-supplemented-sketches/blob/main/examples/catsynth/experiment/run_experiment.py}{\texttt{run\_experiment.py}} \\
Codex App Server adapter & \href{https://github.com/open-horizon-labs/counterexample-supplemented-sketches/blob/main/examples/catsynth/catsynth/codex_app_server.py}{\texttt{codex\_app\_server.py}} \\
OpenAI-compatible adapter & \href{https://github.com/open-horizon-labs/counterexample-supplemented-sketches/blob/main/examples/catsynth/catsynth/openai_compat.py}{\texttt{openai\_compat.py}} \\
Oracle A & \href{https://github.com/open-horizon-labs/counterexample-supplemented-sketches/blob/main/examples/catsynth/catsynth/oracle_a.py}{\texttt{oracle\_a.py}} \\
Oracle B & \href{https://github.com/open-horizon-labs/counterexample-supplemented-sketches/blob/main/examples/catsynth/catsynth/oracle_b.py}{\texttt{oracle\_b.py}} \\
UI fixtures & \href{https://github.com/open-horizon-labs/counterexample-supplemented-sketches/blob/main/examples/catsynth/catsynth/seed.py}{\texttt{seed.py}} \\
UI gate & \href{https://github.com/open-horizon-labs/counterexample-supplemented-sketches/blob/main/examples/catsynth/catsynth/gate.py}{\texttt{gate.py}} \\
UI app & \href{https://github.com/open-horizon-labs/counterexample-supplemented-sketches/blob/main/examples/catsynth/catsynth/app.py}{\texttt{app.py}} and \href{https://github.com/open-horizon-labs/counterexample-supplemented-sketches/tree/main/examples/catsynth/catsynth/static}{\texttt{static/}} \\
Comparison harness & \href{https://github.com/open-horizon-labs/counterexample-supplemented-sketches/blob/main/examples/catsynth/experiment/adaptive_open_world_experiment.py}{\texttt{adaptive\_open\_world\_experiment.py}} \\
Captured run & \href{https://github.com/open-horizon-labs/counterexample-supplemented-sketches/tree/main/examples/catsynth/experiment/results/gpt-5.4-mini-adaptive-open-world-v2-20260712}{Published run} \\
\end{longtable}
}

\section{Claim boundary}\label{claim-boundary}

CatSynth\textquotesingle s iterative gate establishes finite-regression correctness only for these fixtures,
expected outputs, and evaluators. The withheld cases provide limited evidence beyond the gate.
Neither result establishes correctness for all owner profiles, real breed facts, bad expected
outputs or checkers, unencoded policy, future model behavior, or other models and reveal orders.
The capture also does not evaluate the method\textquotesingle s separate sketch-review step. CatSynth preserves
the evidence needed to inspect that boundary.